\documentclass[a4paper,10pt]{article}

\usepackage{authblk}

\usepackage[cmex10,intlimits]{amsmath} 
\usepackage{amsthm,amssymb}
\usepackage{xcolor}
\usepackage{color}
\usepackage{graphicx}
\usepackage{algorithmic}
\usepackage{algorithm}

\usepackage{tikz} 

\usetikzlibrary{positioning}
\usetikzlibrary{calc}
\usepackage{pgfplots}

\newtheorem{theorem}{Theorem}

\newtheorem{lemma}{Lemma}
\newtheorem{problem}{Problem}
\newtheorem{assumption}{Assumption}
\newtheorem{corollary}{Corollary}

\newcommand{\ie}{{\it i.e.},\ }

\DeclareMathOperator*{\argmin}{\arg\!\min}

\newcommand{\Pe}{P_{\mathrm{e}}}
\newcommand{\Pei}{P_{\mathrm{e}}(i)}

\newcommand{\JG}[1]{#1}
\newcommand{\JGb}[1]{#1}

\begin{document}

%
%
%
\title{Optimization of Caching Devices with Geometric Constraints}

\author[a]{Konstantin Avrachenkov}
\author[b]{Xinwei Bai}
\author[b]{Jasper Goseling}

\affil[a]{INRIA Sophia Antipolis, France}
\affil[b]{Stochastic Operations Research, University of Twente,~The~Netherlands}
\affil[{ }]{k.avrachenkov@sophia.inria.fr, x.bai@utwente.nl, j.goseling@utwente.nl}

\maketitle

%
%
%
\begin{abstract}
It has been recently advocated that in large communication systems it is beneficial both for the users and for the network as a whole  to store content closer to users. One particular implementation of such an approach is to co-locate caches with wireless base stations. In this paper we study geographically distributed caching of a fixed collection of files. We model cache placement with the help of stochastic geometry and optimize the allocation of storage capacity among files in order to minimize the cache miss probability. We consider both per cache capacity constraints as well as an average capacity constraint over all caches. The case of per cache capacity constraints can be efficiently solved using dynamic programming, whereas the case of the average constraint leads to a convex optimization problem. We demonstrate that the average constraint leads to significantly smaller cache miss probability. Finally, we suggest a simple LRU-based policy for geographically distributed caching and show that its performance is close to the optimal.
\end{abstract}

%
%
%
\section{Introduction}
\label{sec:introduction}

%
%

We consider caching of a collection of files by a set of geographically distributed storage devices with wireless communications
capabilities and random network coding. Clients can retrieve cached data from all devices that are within its connectivity radius.
Since the caching devices have limited storage capacity, not all files can be stored in all caches. Therefore, there is a positive
probability that a file that is requested by a client cannot be retrieved from the caching devices that are within range and
a cache miss occurs. The general aim of this paper is to optimize the cache allocation so to minimize the cache miss probability.

It has been recently advocated that in large communication systems it is beneficial both for the users and for the network
as a whole to store content closer to users. This idea can be realized by Information Centric Networking (ICN),
a new paradigm for the network architecture where the data is addressed by its name or content directly rather than by its physical location. There is no predefined location for the data in ICN and the content is naturally cached along the retrieval path. Examples of the ICN architecture are CCN/NDN \cite{Jacobson:2009}, DONA \cite{Koponen:2007} and TRIAD \cite{Gritter:2001}. Our results can be useful for the design of the wireless networks with the ICN architecture in which case cellular base stations also serve as caches.
Wireless sensor networks represent another potential application of our results. Sensors have severe limitation on both memory
and transmission capability. It might be useful for sensors to have access to some aggregated characteristics in addition to the local ones. In such a case, our results provide optimal distributed allocation of the aggregated characteristics.

%
%
%

Let us elaborate on the problem formulation in further details. Storage (or caching) devices are placed in the plane according to a homogeneous spatial Poisson process. The homogeneous spatial Poisson process is accepted for modelling the location of base stations
providing a good compromise between realistic representation of the wireless network and mathematical tractability \cite{baccelli1,baccelli2,haenggi2009}.
For some cases, e.g., for Sydney base station network \cite{lee2013}, it has been shown that the spatial homogeneous Poisson process represents very well the distribution of base stations.
In other cases, a non-homogeneous Poisson process can be more appropriate for modelling the distribution
of base stations. In fact, some results of the present work can be extended to the case of non-homogeneous
Poisson process and we discuss such extensions later in the paper.
The size of the file catalog is finite and fixed. A client will request one of the files from the catalog at random according to a known file popularity distribution that is the same for all clients. In particular, for numerical illustration purpose we will consider the case that file popularities follow a Zipf distribution. For the sake of tractable performance evaluation analysis, we make a technical assumption that files consist of the same number of chunks of a fixed size. We suggest to use random linear network coding, in which case linear combination of chunks can be stored in the caching devices. As was shown in~\cite{altman2013coding},
the network coding based allocation strategy outperforms a strategy without coding for a wide range of performance measures
and any spatial distribution of caches.

Our interest in the current paper is in the case when the caches are reachable only within a fixed distance to the client.
This is a standard model in wireless networks which gives high level but still quite accurate representation
of a wireless connection \cite{baccelli1,baccelli2}.
Our goal is to minimize the cache miss probability, which is the probability that a client cannot get the requested file
from the caches within range. Since the probability of not recovering a file from coded chunks is negligible in comparison to
the overall cache miss probability, we concentrate solely on the calculation of the cache miss probability
and on the optimization of the system with respect to this metric.

We have multiple files and a limited memory in each storage device. Thus, the question is how many linear combinations
of each file to store in a particular storage device. Initially, we consider the case when we make the same allocation in all caches,
\ie each cache stores the same number of linear combinations of each file.  As a consequence we guarantee a capacity constraint
on each individual cache. We formulate an optimization problem with a non-convex objective function and linear constraints.
We demonstrate that this problem is a generalization of an unbounded knapsack problem~\cite{martello1990knapsack}. In particular
it is a separable nonlinear integer program, which can be solved using dynamic programming. In addition to providing
a formal statement of this result, we give exact closed form results for some special cases of the problem as well as insight
into the structure of the solution in the general case.

The above formulation leads to the same allocation in each storage device, which likely leads to inefficient memory utilization and
to the lack of file diversity. Thus, we then turn our attention to a relaxation of the problem in which, instead of imposing
a hard capacity constraint on each of the caches, we require that the average storage space used in the caches is upper bounded.
In particular, we consider cache allocation strategies in which the number of linear combination to store for a file in a cache
is a random variable. The number of such combinations is independently and identically decided for each cache. We impose an average
capacity constraint on the number of chunks stored in a caching device, where the average is over the caching devices. We analyze
the resulting optimal strategy for the case when files consist of a single piece and show that the performance under an average capacity significantly outperforms the optimal performance under a per cache capacity constraint.

Finally, we consider a dynamic scenario when the clients arrive over time. We study two LRU-based caching policies, cooperative
and fully distributed. Both policies demonstrate \JGb{that} performance \JGb{is} not far from the optimal one and \JGb{that} there is a small loss of efficiency
in the fully distributed case compared to the cooperative case. This indicates that a simple distributed LRU-based caching policy
can be safely deployed in practice for geographically distributed caches. \JGb{Also, it indicates that our results on the optimal placement policies can provide insight into the performance in the dynamic setting.  }

Let us outline the organization of the paper. In Section~\ref{sec:model} we define the model, discuss the constraints and
optimization criterion. The problem with per cache constraints is analysed in Section~\ref{sec:analysis}. In particular,
we provide structural insight into the optimal storage allocation strategy and show that the problem is a generalization
of the unbounded knapsack problem and can be solved by dynamic programming approach. Then, in Section~~\ref{sec:average}
we introduce the average constraint, which makes memory usage more efficient and increases file diversity. In an important
particular case we are able to solve the average constraint problem in a closed form.
In Section~\ref{sec:dynamic} we present distributed and cooperative LRU-based policies. In Section~\ref{sec:numerical}
we demonstrate that the performance of the distributed LRU-based policy is not far from the optimal performance.
The numerical results of Section~\ref{sec:numerical} also confirm that the average constraint in comparison with
the per cache constraint, brings improved efficiency and file diversity.
Finally, in Section~\ref{sec:disc} we provide a discussion of our result and an outlook on future research.

%
%
%
\section{Related work}
\label{sec:related}

Literature on caching is vast. Therefore, we limit our discussion to work on caching which we feel is most relevant to the present work. 
The application of network coding for distributed storage is studied in~\cite{dimakis2010network} and in~\cite{golrezaei13commmag,ji2013throughput} specifically for the case of content distribution in wireless networks. The use of  coding was also explored in~\cite{dimakis2005ubiquitous} where it was shown how to efficiently allocate the data at  caches with the aim of ensuring that any sufficiently large subset of caches can provide the complete data. The difference with the current work is that we are taking the geometry of the deployment of the storage devices into account.
In~\cite{maddah2014fundamental}, see also \cite{pedarsani2014online},
coding strategies for networks of caches are presented, where each user has access to a single cache and a direct link to the source. It is demonstrated how coding helps to reduce the load on the link between the caches and the source. Note that we assume that different transmissions from caches to the clients are orthogonal, for instance by separating them in time or frequency. In~\cite{6225433} the impact of non-orthogonal transmissions is considered and scaling results are derived on the best achievable transmission rates. In~\cite{hachem2014coded} a heterogeneous system of small coverage access points and large coverage base stations is considered.

Systems of distributed storage devices or caches can be classified according to the amount of coordination between the devices. In~\cite{rosensweig09infocom} an approach with implicit coordination is proposed. Networks of caches are notoriously difficult to analyze. Only some very particular topologies and caching strategies (see \JGb{\cite{fofack2014performance}} and references therein) or approximations \cite{Rosensweig2010,Che:2006} have been studied. In a recent work \cite{Rosensweig2013} ergodicity of cache networks has been investigated. Using continuous geometrical constraints on cache placement instead of combinatorial constraints allows us to obtain exact analytical results.

Other work on caching in wireless networks is, for instance, \cite{nuggehalli2003energy, jin2005content, yin2006supporting}. In \cite{nuggehalli2003energy} the authors analyze the trade-off between energy consumption and the retrieval delay of data from the caches. In \cite{jin2005content}, the authors consider the optimal number of replicas of data such that the distance between a requesting node and the nearest replica is minimized. Data sharing among multiple caches  such that the bandwidth consumption and the data retrieval delay are minimal is considered in  \cite{yin2006supporting}. None of~\cite{nuggehalli2003energy, jin2005content, yin2006supporting} are considering coded caching strategies.

We would like to emphasize that except for~\cite{golrezaei13commmag,ji2013throughput} none of the above mentioned works considered continuous geometric constraints on storage device placement.

Networks of wireless caches in the plane, \ie with geometric constraints, were first studied in~\cite{altman2013coding} for the case of single file. The tradeoff between the retrieval performance and the deployment cost in terms of number of caches and their capacity was studied in~ \cite{mitici2013optimal}. Both papers considered the storage of just one single data file.
In \cite{bastug2014cache,bastug2014social} the framework of stochastic geometry was applied to performance evaluation
of a network of small base stations with emphasis on physical layer.
The question of optimal storage allocation was not studied there.
In the recent work \cite{BG15} a model very similar to our average capacity constraint model has been
analysed from a different angle. We would like to note that the present work has been done before \cite{BG15}.



%
%
%
\section{Model and Notation}
\label{sec:model}

Caching devices are placed in the plane $\mathbb{R}^2$ according to a homogeneous spatial Poisson process with density $\lambda$.
The spatial Poisson process is known to be an appropriate generic model for location of base stations or sensors in wireless
networks \cite{baccelli1,baccelli2,haenggi2009,lee2013}.
The devices serve as caches for a catalog of $L$ data files. Without loss of generality, we consider a single client
that is located at an arbitrary location in the plane \cite{baccelli1,baccelli2} and can access only caches within radius $r$.
Since the caches are distributed according to a homogeneous spatial Poisson process, the number of caches within radius $r$
follows a Poisson distribution with parameter $x=\lambda\pi r^2$. That is,
$$
P(n\ {\rm caches\ within\ radius}\ r)=\frac{x^n}{n!}e^{-x}.
$$
The parameter $x$ has an interpretation as an expected number of devices inside the area of size $\pi r^2$.
The client is interested in retrieving one of the $L$ files. The file that is required by the client is selected at random. The probability that the $i$-th file is selected is $p_i$, $i=1,\dots,L$. Without loss of generality, we assume that $p_1\ge p_2\ge\ldots \ge p_L$. The probability distribution $p_i$ represents the popularity of the files. Most of the results in this paper will be obtained for an arbitrary file popularity distribution. In some cases, in particular for illustration of our results by numerical examples, we consider the Zipf distribution. Let $p_i^z$ denote the probability of file $i$ under a Zipf distribution~\cite{mahanti2013tale} with parameter $s>0$, \ie $p_i^z=i^{-s}/\sum_{k=1}^Lk^{-s}$.

We suppose that files consist of $N$ chunks (ICN terminology). For the sake of tractability, we assume that all packets of all files are of the same size. Each caching device can store at most $C$ packets. We assume that $C<LN$, \ie we cannot store all files in a device. Therefore, a \JGb{means} of allocating (parts of) files to caches needs to be devised. Inevitably, for any allocation strategy there will be a positive probability that the client cannot retrieve the desired file from the caches within its range. Our interest in this paper is \JGb{in} minimizing this probability, the cache miss probability, by optimizing the storage strategy. We allow for caches to store only part of a file. Also, we allow for random linear network coding to be used. As a consequence, caches do not store the data packets themselves, but store instead one or more random linear combinations of the data packets of a file. The purpose of using network coding is that with high probability in order to recover a file it is sufficient for the client to retrieve any $N$ linear combinations of packets. In a network of caches the probability that the client cannot recover a file from $N$ linearly coded packets is negligible compared to the overall
cache miss probability \cite{altman2013coding,ncfundamentals}. Therefore, we ignore this event in this paper and use the following assumption.
\begin{assumption} \label{ass:suff}
A file can be recovered from any set of $N$ linear combinations of packets from that file.
\end{assumption}

The storage strategy is based on storing in each caching device $n_i$ linear combinations of the packets of file $i$.
Since the considered model is space homogeneous, all caching devices follow the same caching strategy.
The capacity constraint that we need to satisfy is
\begin{equation}
\sum_{i=1}^L n_i \leq C.
\end{equation}
Now in order to retrieve file $i$ the client needs to obtain at least $N$ linear combinations for that file. If the caches within radius $r$ cannot provide these linear combinations, a cache miss occurs. Our interest in this paper is in minimizing the cache miss probility by optimizing the values $n_i$, $i=1,\dots,L$. The probability is over the placement of the caches as well as the selection of the file by the client. More precisely, our performance measure  of interest is
\begin{equation}
 \Pe= \sum_{i=1}^L p_i \Pei,
\end{equation}
where
\begin{multline}
\Pei = \Pr\{\text{client cannot obtain file } i \text{ from caches within distance } r\}.
\end{multline}

In this section we have defined only the problem with per cache capacity constraints. The relaxation to average constraints
is defined and analysed in Section~\ref{sec:average}. In the next section we first analyze the case of per cache capacity constraints.
Then, in Section~\ref{sec:dynamic} we consider a dynamic scenario with arriving and departing users.

%
%

%
%
%
\section{Individual Cache Capacity Constraints}
\label{sec:analysis}

We start this section with a formulation of the optimization problem in Subsection~\ref{ssec:problem}. Next we provide some results on the structure of the optimal solution to this problem in Subsection~\ref{ssec:structure}. In Subsection~\ref{ssec:n1opt} we give an analytical expression for the optimal solution for the case that files consist of a single chunk. Finally, in Subsection~\ref{ssec:dynprog} we provide a dynamic programming approach for solving the general case.

\subsection{Formulation of optimization problem} \label{ssec:problem}

The client can connect to all caches that are within radius $r$. Since the caches are distributed according to a homogeneous Poisson process, the number of caches within radius $r$ follows a Poisson distribution with parameter $x=\lambda\pi r^2$.
%
When the client wants file $i$, $\lceil N/n_i \rceil$ caches are needed to get the complete file. This request will be missed if there are less than $\lceil N/n_i \rceil$ caches within radius $r$ to the client. Therefore, the miss probability for file $i$ is given by
\begin{eqnarray}
\nonumber\Pei&=&\sum_{k=0}^{\lceil N/n_i \rceil-1}P(k\ {\rm caches\ within\ radius}\ r)\\
\nonumber &=&\sum_{k=0}^{\lceil N/n_i \rceil-1}\frac{x^k}{k!}e^{-x}\\
&=& Q(\lceil N/n_i \rceil,x),
\end{eqnarray}
where $Q$ is the regularized incomplete Gamma function.
Since file $i$ is requested with probability $p_i$, the expected miss probability is
$$
\Pe = \sum_{i=1}^L p_iQ(\lceil N/n_i\rceil,x).
$$
For notational convenience, let the function $f$ be defined as
\begin{equation} \label{eq:f}
f(n_i)=Q(\lceil N/n_i\rceil,x).
\end{equation}
From the above it follows that the minimization of the cache miss probability $P_e$ is given by the following optimization problem.
\begin{problem}\label{pb:opt}
\begin{eqnarray}\label{opt}
\nonumber {\rm min} & &\sum_{i=1}^L p_if(n_i)\\
\nonumber {\rm subject\  to} & &\sum_{i=1}^L n_i=C,\\
 & & n_i\in \mathbb{N}, \ i=1,\dots,L.
\end{eqnarray}
\end{problem}
Note that the objective function of this optimization problem is non-increasing in $n_i$. However, for $N>1$ it is not convex. Since the objective function is non-increasing, we consider only equality $\sum_{i=1}^L n_i=C$ in the capacity constraint.

\subsection{Structure of the optimal solution} \label{ssec:structure}
Our first result deals with the structure of the optimal solution. In particular we demonstrate that the number of linear combinations stored for file $i$ is a non-increasing function in $i$.
\begin{theorem} \label{th:structure}
Let $\bar{n}=(\bar n_1,\dots,\bar n_L)$ be an optimal solution to Problem~\ref{pb:opt}. Then $\bar n_1\ge \bar n_2\ge \dots \ge \bar n_L$.
\end{theorem}
\begin{proof}
Suppose there exists $j<k$ for which $n_j<n_k$, then consider $n'$, constructed by having $n_j'=n_k$, $n_k'=n_j$ and the others remain the same, then we can get
\begin{align*}
\sum_{i=1}^L p_if(n'_i) & -\sum_{i=1}^L p_if(n_i) \\
&= p_jf(n_j')+p_kf(n_k')-p_jf(n_j)-p_kf(n_k) \\
&= p_jf(n_k)+p_kf(n_j)-p_jf(n_j)-p_kf(n_k)\\
&= (p_j-p_k)[f(n_k)-f(n_j)].
\end{align*}
Since $j<k$, then $p_j-p_k\ge 0$. Also, since $f$ is non-increasing in $n_i$ and $n_j< n_k$, then we can get that $f(n_j)\ge f(n_k)$, \ie $f(n_k)-f(n_j)\le 0$. Therefore, $f(n')-f(n)\le 0$, which means that after the exchange, the objective value will not become higher.
Then we can keep doing exchange until $n_1\ge n_2\ge \ldots \ge n_L$.
\end{proof}

As we can see, if \JGb{a} per cache constraint is used, the optimal allocation is the same for all caches. This will likely result in
inefficient memory usage and the total absence of some files from the caching system. In Section~\ref{sec:average} we introduce the
average capacity constraint which will help to mitigate these issues.

\subsection{Optimal solution for $N=1$} \label{ssec:n1opt}
Next, we consider the case that files consist of a single chunk, \ie $N=1$. This implies that we either store a file completely in each cache, or not at all, \ie $n_i$ can be either 0 or 1. If $n_i=1$, then file $i$ is stored in every cache.
In this case when a client requests file $i$, the miss probability will be
$$
\Pei=e^{-x}.
$$
If $n_i=0$, the miss probability will be 1 if it is requested. Therefore, we can see that
\begin{eqnarray}\label{mpi}
\Pei=\left\{
 \begin{array}{c@{,\quad}c}
 e^{-x} & {\rm if}\  n_i=1,\\
 1 & {\rm if}\  n_i=0\\
 \end{array}\right.
\end{eqnarray}
and we can write equation \eqref{mpi} as
\begin{eqnarray}
\Pei=e^{-n_i x}.
\end{eqnarray}

For the special case of $N=1$, the general optimization problem, Problem~\ref{pb:opt}, reduces to the following problem.
\begin{problem}\label{pb:discrete}
\begin{eqnarray*}
\nonumber {\rm min} & &\sum_{i=1}^L p_ie^{-n_i x}\\
\nonumber {\rm subject\  to} & &\sum_{i=1}^L n_i=C\\
 & & n_i\in\{0,1\},\ i=1,\dots,L.
\end{eqnarray*}
\end{problem}

Since $n_i$ is binary and we know that the optimal solution has a structure $n_1\ge n_2\ge... n_L$, it follows directly from Theorem~\ref{th:structure} that the optimal solution of Problem~\ref{pb:discrete} is as stated in the following result.
\begin{corollary}
The optimal solution of Problem \ref{pb:discrete} is $\bar n=(\bar n_1, \bar n_2,..., \bar n_L)$, where
\begin{eqnarray*}
\bar n_i=\left\{
\begin{array}{c@{, \ }c}
1 & {\rm if\ } i\le C,\\
0 & {\rm if\ } i>C.
\end{array}\right.
\end{eqnarray*}
\end{corollary}

Note that contrary to the case $N>1$ the objective function of Problem~\ref{pb:discrete} is convex. We will make use of this property in Section~\ref{sec:average}, where we will revisit the case $N=1$ under an average capacity constraint.

%


\subsection{Dynamic Programming} \label{ssec:dynprog}

In this section we return to the general case of arbitrary $N$. As already discussed in Section~\ref{sec:introduction}, Problem~\ref{pb:opt} is a generalization of the unbounded knapsack problem. The generalization comes from the fact that the objective function is not a weighted sum of the variables $n_i$, but a non-convex function \JGb{in} these variables. It is well-known that the unbounded knapsack problem can be solved in pseudo-polynomial time using dynamic programming~\cite{martello1990knapsack}. In this section we demonstrate that Problem~\ref{pb:opt} can also be solved using dynamic programming.

In order to formulate a dynamic programming solution we interpret Problem~\ref{pb:opt} as follows. We have $C$ units in total, and there are $L$ slots to put the units in. Assigning $n_i$ units to slot $i$ induces a certain cost. Our goal is to distribute all of the $C$ units over these slots with a minimal total cost, which is defined as $\sum_{i=1}^L p_if(n_i)$, where $f$ is defined in~\eqref{eq:f}. The idea of dynamic programming is to assign the units one by one, leading to a recursive procedure in both $L$ and $C$.

More precisely, consider the problem
\begin{eqnarray}\label{dynaprob}
\nonumber {\rm min} & &\sum_{i=1}^\ell p_if(n_i),\\
{\rm subject\  to} & &\sum_{i=1}^\ell n_i=c,
\end{eqnarray}
and let $F(\ell,c)$ denote its optimal value. Our interest is in $F(L,C)$ and $\bar n$, a solution attaining $F(L,C)$. For $\ell=2,\dots,L$ and $c=0,\dots,C$ let
\begin{equation}
\tilde n_{\ell,c} = \argmin_{n\in \{0,...,c\wedge N\}} \left\{F(\ell-1,c-n)+p_{\ell}f(n)\right\},
\end{equation}
where $c\wedge N=\min\{c,N\}$. The dynamic programming approach to Problem~\ref{pb:opt} is based on the observation that for $2\leq\ell\leq L$ and $0\leq c\leq C$ we can express $F(\ell,c)$ as
\begin{equation}\label{dynamic}
\begin{aligned}
F(\ell,c) &= F(\ell-1,c-\tilde n_{\ell,c})+p_{\ell}f(\tilde n_{\ell,c}).  \\
\end{aligned}
\end{equation}
\begin{algorithm}[t]
\begin{algorithmic}
\FOR{$c=0:C$}
\STATE $F(1,c) = \begin{cases} p_1f(c),\quad &\text{if }c\leq N,\\ p_1f(N),\quad &\text{if }c> N.\end{cases}$
\ENDFOR
\FOR{$\ell=2:L$}
\FOR{$c=0:C$}
\STATE $\tilde n_{\ell,c}\!=\!\argmin_{n\in \{0,\dots,N\wedge c\}} \!\!\left\{F(\ell-1,c-n)+p_{\ell}f(n)\right\}\!,$
\STATE $F(\ell,c) = F(\ell-1,c-\tilde n_{\ell,c})+p_{\ell}f(\tilde n_{\ell,c})$.
\ENDFOR
\ENDFOR
\STATE $c=C$
\FOR{$\ell=L:-1:2$}
\STATE $\bar n_\ell = \tilde n_{\ell,c},$
\STATE $c=c-\bar n_\ell.$
\ENDFOR
\STATE $\bar n_1 = c,$
\STATE $P_e = F(L,C)$.
\end{algorithmic}
\caption{Dynamic Programming Algorithm for Problem~\ref{pb:opt}}
\label{alg:dynprog}
\end{algorithm}
%
%
%
%
The procedure is initialized by considering $\ell=1$ and $0\leq c\leq C$, for which we know that the optimal value $F(1,c)=p_1f(c\wedge N)$. Next we apply formula (\ref{dynamic}) iteratively.  After computing all values $F(\ell,c)$, the solution $\bar n$ can be constructed from the values of $\tilde n_{\ell,c}$ by tracking backwards starting at $\ell=L$. The complete procedure is presented as Algorithm~\ref{alg:dynprog}. The following theorem provides a formal statement of the result. The proof follows from standard results on dynamic programming.
\begin{theorem}
Algorithm~\ref{alg:dynprog} provides a globally optimal solution to Problem~\ref{pb:opt} in pseudo-polynomial time.
\end{theorem}

In Section~\ref{sec:numerical} we will provide additional insight into the optimal solution of Problem~\ref{pb:opt}.

%
%
%
\section{Average Capacity Constraints}
\label{sec:average}

%
%

In this section, instead of imposing an individual per cache constraint on each of the devices, we require that the average storage space used in the devices is upper bounded. We analyze the resulting optimal strategy for the case that files consist of a single chunk ($N=1$) and show that the performance under an average capacity constraint significantly \JGb{outperforms} the optimal performance under a per cache capacity constraint.

Since files consist of a single chunk, the choice to make is whether to store the complete file or not to store the file at all.
The proposed strategy places file $i$ in a cache with probability $q_i$. Placement of files is independent between caches. By the independence of the placement over the caches and the thinning property of the Poisson process \cite{baccelli1}, 
it follows that those caches that contain file $i$ are again distributed according to a spatial Poisson process, this time with density $q_i\lambda$. Therefore, the probability that file $i$ cannot be retrieved from the caches within distance $r$ is
\begin{equation}
\Pei = p_ie^{-q_i x},
\end{equation}
with $x=\lambda \pi r^2$.

Now the goal is to optimize $\sum_{i=1}^L \Pei$ subject to the capacity constraint. This leads to the following optimization problem.
\begin{problem}\label{pb:relax}
\begin{eqnarray*}
{\rm min} & & \sum_{i=1}^L p_ie^{-q_i x}\\
{\rm subject\  to} & & \sum_{i=1}^L q_i= C,\\
 & & 0\le q_i\le 1,\ i=1,\dots,L.
\end{eqnarray*}
\end{problem}
Note that contrary to the objective function of Problem~\ref{pb:opt}, the above objective function is convex. 
Also note that in contrast to Problem~2 the variables in Problem~\ref{pb:relax} are continuous. 
Therefore, Problem~\ref{pb:relax} is a convex optimization problem.

\subsection{Optimal solution}
Since Problem~\ref{pb:relax} is convex, the Karush-Kuhn-Tucker (KKT) conditions provide necessary and sufficient conditions for optimality. We will construct an explicit analytical solution to Problem~\ref{pb:relax} that satisfies the KKT conditions.

The Lagrangian function corresponding to Problem~\ref{pb:relax} is
\begin{multline}
L(q,\nu,\lambda,\omega)=\sum_{i=1}^Lp_ie^{-q_i x}+\nu(\sum_{i=1}^L q_i -C) \\
-\sum_{i=1}^L \lambda_i q_i+\sum_{i=1}^L \omega_i (q_i-1),
\end{multline}
where $q, \lambda,\omega \in \mathbb{R}_{+}^L, \nu\in\mathbb{R}$.

Let $\bar q, \bar{\lambda}, \bar{\omega}$ and $\bar{\nu}$ be primal and dual optimal.
Then the KKT conditions for Problem~\ref{pb:relax} state that
\begin{eqnarray}
\label{eq1} 0\le \bar{q_i}&\le& 1,\\
\label{eq2} \sum_{i=1}^L \bar{q_i}&=&C,\\
\label{eq3} \bar{\lambda_i}&\ge& 0,\ \forall i=1,...,L,\\
\label{eq4} \bar{\omega_i}&\ge& 0,\ \forall i=1,...,L,\\
\label{eq5} \bar{\lambda_i}\bar{q_i}&=&0,\ \forall i=1,...,L,\\
\label{eq6} \bar{\omega_i}(\bar{q_i}-1)&=&0,\ \forall i=1,...,L,\\
\label{eq7} -p_ixe^{-\bar{q_i}x}+\bar{\nu}-\bar{\lambda_i}+\bar{\omega_i}&=&0,\ \forall i=1,...,L.
\end{eqnarray}
For notational convenience we introduce the functions $g_i:\mathbb{R}\to[0,1]$, $i=1,\dots,L,$ as follows
\begin{eqnarray}\label{eq:Ndef}
g_i(\nu)=\left\{
\begin{array}{r@{,\quad}l}
\displaystyle 1 & {\rm if\ } \nu\le p_ixe^{-x},\\
\displaystyle \frac{1}{x}\log\frac{p_ix}{\nu} & {\rm if\ } p_ixe^{-x}<\nu<p_ix, \\
\displaystyle 0 & {\rm if\ } \nu\ge p_ix.
\end{array}\right.
\end{eqnarray}
Furthermore, let $g:\mathbb{R}\to[0,L]$ be defined as $g(\nu)=\sum_{i=1}^L g_i(\nu)$. Observe that $g(\nu)=L$ for $\nu\in (-\infty, p_Lx e^{-x}]$, that $g(\nu)=0$ for $\nu\in [p_1x, \infty)$ and that it is strictly decreasing in the interval $(p_Lx e^{-x}, p_1x)$.

\begin{lemma} \label{lem:structure}
Let $\bar q$ and $\bar \nu$ be optimal. Then $\bar q = (g_1(\bar\nu),\dots,g_L(\bar\nu))$.
\end{lemma}
\begin{proof}
Let $i\in\{1,\dots,L\}$. From~\eqref{eq5}, \eqref{eq6} and \eqref{eq7}, we have
\begin{equation}
\bar{\omega_i}=\bar{q_i}(p_ixe^{-\bar{q_i}x}-\bar{\nu}),
\end{equation}
which, when inserted into \eqref{eq6}, gives
\begin{equation}\label{main}
\bar{q_i}(\bar{q_i}-1)(p_ixe^{-\bar{q_i}x}-\bar{\nu})=0.
\end{equation}
From~\eqref{main}, we see that $0<\bar{q_i}<1$ only if $\bar{\nu}=p_ixe^{-\bar{q_i}x}$. Since $0\leq q_i\leq 1$, this implies that $\nu\in[p_ixe^{-x}, p_ix]$.

If $\bar{\nu}< p_ixe^{-x}$, then
$$
\bar{\omega_i}=\bar{\lambda_i}+p_ixe^{-\bar{q_i}x}-\bar{\nu}>0.
$$
Thus, from \eqref{eq6}, we have $q_i=1$. Similarly, if $\bar{\nu}> p_ix$, it follows from
$$
\bar{\lambda_i}=\bar{\omega_i}+\bar{\nu}-p_ixe^{-\bar{q_i}x}>0
$$
and~\eqref{eq5} that $\bar{q_i}=0$.
\end{proof}

It remains to solve for $\bar\nu$. The complete solution is provided by the next theorem.


\begin{theorem} \label{th:avg}
The optimal solution of Problem~\ref{pb:relax} is given by
\begin{eqnarray*}
\bar{q_i}=\left\{
\begin{array}{c@{,\quad}l}
\displaystyle 1 & {\rm if\ } i<k_1,\\
\displaystyle \frac{1}{x}\log\frac{p_ix}{\bar\nu} & {\rm if\ } k_1\le i\le k_2,\\
\displaystyle 0 & {\rm if\ } i>k_2
\end{array}\right.
\end{eqnarray*}
where $k_1, k_2$ are given by
\begin{gather}
k_1 = \min\big\{ 1\leq\ell\leq L \big|\ g(p_\ell xe^{-x})\geq C \big\}, \label{eq:k1}\\
k_2 = \max\big\{ 1\leq\ell\leq L \big|\ g(p_\ell x)\leq C \big\} \label{eq:k2}
\end{gather}
and
\begin{equation}\label{nu}
\bar{\nu}=\exp \left\{\frac{1}{k_2-k_1+1} \sum_{j=k_1}^{k_2} \log p_jx-\frac{x(C-k_1+1)}{k_2-k_1+1}\right\}.
\end{equation}
\end{theorem}
\begin{proof}
From Lemma~\ref{lem:structure} it follows that there exist $k_1, k_2\in [1,L]$ such that $\bar q_1=\bar q_2=\dots=\bar q_{k_1-1}=1$ and $\bar q_{k_2+1}=\bar q_{k_2+2}=\dots=\bar q_L=0$. In particular, $k_1$ is given by
\begin{equation} \label{eq:k1prelim}
k_1 = \min\big\{ 1\leq\ell\leq L \big|\ \bar v>p_\ell xe^{-x} \big\}.
\end{equation}
Note that the above minimum is guaranteed to exist, because otherwise $q_i=1$ for all $i=1,\dots,L$, leading to a contradiction on the assumption that $\sum_{i=1}^L q_i=C<L$. Condition~\eqref{eq:k1} is obtained by applying the non-increasing function $g$ to the LHS and the RHS in the constraint in~\eqref{eq:k1prelim} and by observing that  from Lemma~\ref{lem:structure} and~\eqref{eq2} it follows that $g(\bar\nu)=C$.

Condition~\eqref{eq:k2} follows in similar lines by starting from
\begin{equation}
k_2 = \max\big\{ 1\leq\ell\leq L \big|\ \bar v< p_\ell x\big\}.
\end{equation}
This maximum exists, because otherwise $\sum_{i=1}^L q_i=0$, which contradicts~\eqref{eq2}.
Finally, the proof of the lemma is completed by solving for $\bar\nu$ in $g(\bar\nu)=C$.
\end{proof}

We note that in contrast to the solution of Problem~\ref{pb:opt} the solution in the case of the average capacity
constraint admits \JGb{a} probabilistic policy of file placement. This should improve the system efficiency as well as
file diversity. We will further illustrate the results of Theorem~\ref{th:avg} in Section~\ref{sec:numerical}.

%
%
%
\section{Dynamic Setting}
\label{sec:dynamic}
In this section we consider a dynamic scenario in which clients arrive over time. More precisely, we consider $L>1$ files, each of $N=1$ packets.  Clients arrive over time. We assume that at any time there is at most one client, \ie we assume that the request of a user is completely handled and that the caches are updated before the next client arrives. This assumption is just for
simplicity of modelling and could be safely neglected in \JGb{a} real implementation. The files that are requested by users are selected at random according to a Zipf distribution, independently across users. Clients arrive to random locations in the plane. As in the other parts of this paper, a client can connect to all caches that are within range $r$ of the client.

The caching policy is as follows. If the file requested by the user is present in any of the caches that are within its range, the file is delivered to the client from the cache. If the requested file is not present in the cache it is fetched from the server and delivered to the user. In addition, the file is then placed in the cache that is closest to the user.

Each cache individually follows the Least Recently Used (LRU) policy for caching files. This means that each cache keeps an ordered list of the files that are locally cached. If a file is served to a client it is moved to the head of the list. If a file was fetched from the server it will be placed at the head of the list in the cache that is closest to the user that is requesting that file. If the number of files in the list is exceeding the cache capacity the file at the tail of the list is dropped from the cache.

Note that in the caching policy as described above there is no cooperation between caches. It is a straightforward extension of the LRU policy to a network of caches. In particular, the caching policy does not make use of information about the file popularity. In this section we are interested in comparing the performance of this very simple policy with the optimal allocation strategy of Section~\ref{sec:average}. In order to do so we have simulated the LRU policy as described above. The numerical results are presented in Section~\ref{sec:numerical}.

In addition to the comparison with the optimal allocation strategy we consider the performance of an LRU policy in which the caches fully cooperate. More precisely, we analyze the performance of a single LRU cache with a capacity that equals the expected sum capacity of all caches that are within range of a client. This allows us to evaluate the `penalty to pay' for distributing cache capacity over several caches that operate independently. The expected number of caches that are within range of a client is $\lambda\pi r^2$. The expected sum capacity is therefore $C \lambda\pi r^2$. The cache miss probability of a single LRU cache is known to be accurately approximated with the Che approximation~\cite{che2002hierarchical,fricker2012versatile}. We provide the numerical evaluation of this approximation for a cache of capacity $C\lambda\pi r^2$ in Section~\ref{sec:numerical}.

%
%
%
\section{Numerical Evaluation}
\label{sec:numerical}

%
%

In this section we present a numerical evalution of the results obtained in the previous sections of this paper. In particular, we consider the case of file popularities following a Zipf distribution, \ie $p_i^z=i^{-s}/\sum_{k=1}^Lk^{-s}$, with parameter $s$. The numerical illustrations will provide some additional insights into the behavior of the optimal cache allocation policies as well as into the behavior of the proposed LRU strategies. In particular, we compare cooperative and fully distributed LRU-based
caching policies.

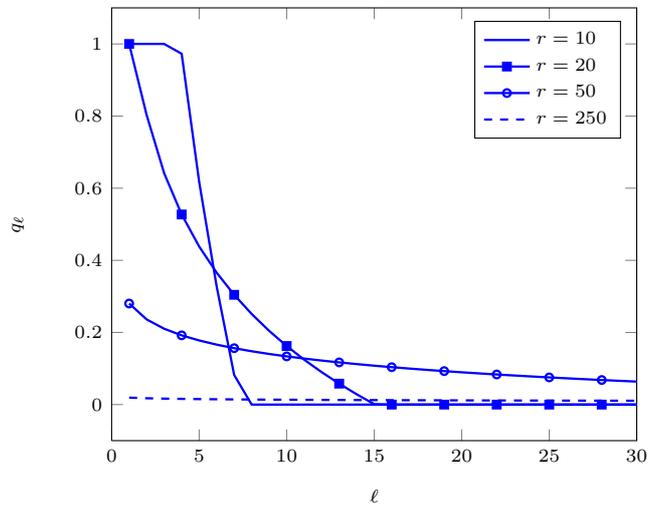
\begin{figure}
\centering
\begin{tikzpicture}
\begin{axis}[
  width=.7\linewidth,
  xlabel=$\ell$,ylabel=$q_\ell$, 
 xmin=0, xmax=30,
  font=\scriptsize,
  legend style={
        cells={anchor=west},
        legend pos=north east,
       font=\scriptsize,
    }
]

\addplot[
  line width=.3mm,color=blue,
  mark=none
  ]
table[
  header=false,x index=0,y index=1,
  ]
{matlab_figures/soft1.csv};
\addlegendentry{$r=10$};

\addplot[
  line width=.3mm,color=blue,
  mark=square*,mark repeat=3,mark phase=0,mark size=.5mm,mark options={solid}
  ]
table[
  header=false,x index=0,y index=2,
  ]
{matlab_figures/soft1.csv};
\addlegendentry{$r=20$};

\addplot[
  line width=.3mm,color=blue,
  mark=o,mark repeat=3,mark phase=0,mark size=.5mm
  ]
table[
  header=false,x index=0,y index=3,
  ]
{matlab_figures/soft1.csv};
\addlegendentry{$r=50$};

\addplot[
  line width=.3mm,color=blue,dashed,
  mark=none
  ]
table[
  header=false,x index=0,y index=4,
  ]
{matlab_figures/soft1.csv};
\addlegendentry{$r=250$};

\end{axis}
\end{tikzpicture}
\caption{Optimal allocation probabilities under an average capacity constraint. ($L=2000$, $N=1$, $C=5$, $\lambda=2\cdot 10^{-3}$, $s=1$)}
\label{fig:avgzipf2a}
\end{figure}

\begin{figure}
\centering
\begin{tikzpicture}
\begin{axis}[
  width=.7\linewidth,
  xlabel=$\ell$,ylabel=$q_\ell$, 
xmin=0, xmax=30,
  font=\scriptsize,
  legend style={
        cells={anchor=west},
        legend pos=north east,
       font=\scriptsize,
    }
]

\addplot[
  line width=.3mm,color=blue,
  mark=none
  ]
table[
  header=false,x index=0,y index=1,
  ]
{matlab_figures/soft2.csv};
\addlegendentry{$r=10$};

\addplot[
  line width=.3mm,color=blue,
  mark=square*,mark repeat=3,mark phase=0,mark size=.5mm,mark options={solid}
  ]
table[
  header=false,x index=0,y index=2,
  ]
{matlab_figures/soft2.csv};
\addlegendentry{$r=20$};

\addplot[
  line width=.3mm,color=blue,
  mark=o,mark repeat=3,mark phase=0,mark size=.5mm
  ]
table[
  header=false,x index=0,y index=3,
  ]
{matlab_figures/soft2.csv};
\addlegendentry{$r=50$};

\addplot[
  line width=.3mm,color=blue,dashed,
  mark=none
  ]
table[
  header=false,x index=0,y index=4,
  ]
{matlab_figures/soft2.csv};
\addlegendentry{$r=250$};

\end{axis}
\end{tikzpicture}
\caption{Optimal allocation probabilities under an average capacity constraint. ($L=2000$, $N=1$, $C=10$, $\lambda=2\cdot 10^{-3}$, $s=1$)}
\label{fig:avgzipf2b}
\end{figure}
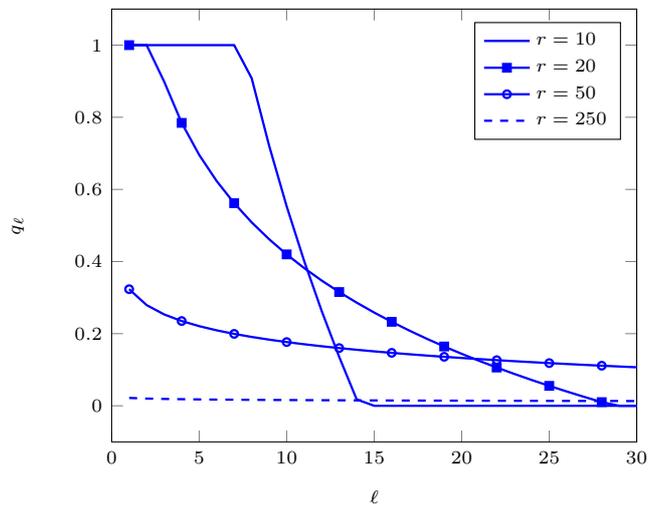

\begin{figure}
\centering
\begin{tikzpicture}
\begin{axis}[
  width=.7\linewidth,
  xlabel=$r$,ylabel=$P_e$, 
 ymin = 0, ymax=1,
  font=\scriptsize,
  legend style={
        cells={anchor=west},
        legend pos=south west,
       font=\scriptsize,
    }
]

\addplot[
  line width=.3mm,color=blue,
  mark=square*,mark repeat=3,mark phase=1,mark size=.5mm,mark options={solid}
  ]
table[
  header=false,x index=0,y index=1,
  ]
{matlab_figures/hardsoftL2000C10.csv};
\addlegendentry{$C=10$, individual};

%
%

\addplot[
  line width=.3mm,color=blue,dashed,
  mark=square*,mark repeat=3,mark phase=2,mark size=.5mm,mark options={solid}
  ]
table[
  header=false,x index=0,y index=1,
  ]
{matlab_figures/hardsoftL2000C50.csv};
\addlegendentry{$C=50$, individual};

\addplot[
  mark=none,line width=.3mm,color=red,
  mark=none
  ]
table[
  header=false,x index=0,y index=2,
  ]
{matlab_figures/hardsoftL2000C10.csv};
\addlegendentry{$C=10$, average};

%
%
%
%
%

\addplot[
  mark=none,line width=.3mm,color=red,dashed,
  mark=none
  ]
table[
  header=false,x index=0,y index=2,
  ]
{matlab_figures/hardsoftL2000C50.csv};
\addlegendentry{$C=50$, average};

\end{axis}
\end{tikzpicture}
\caption{Cache miss probability under individual and average cache capacity constraints. ($L=2000, N=1, \lambda=2\cdot 10^{-3}, s=1$)}
\label{fig:newa}
\end{figure}
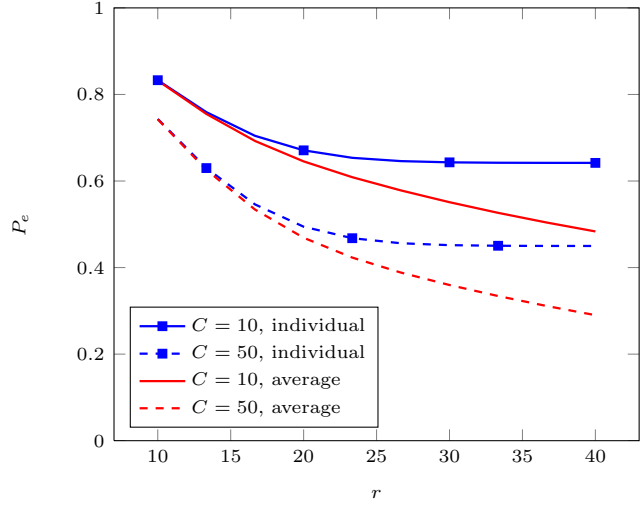

\begin{figure}
\centering
\begin{tikzpicture}
\begin{axis}[
  width=.7\linewidth,
  xlabel=$i$,ylabel=$n_i$, 
 xmin=0, xmax=20,
  font=\scriptsize,
  legend style={
        cells={anchor=west},
        legend pos=north east,
       font=\scriptsize,
    }
]

\addplot[
  only marks,color=red,
  mark=x,mark size=1mm
  ]
table[
  header=false,x index=0,y index=1,
  ]
{matlab_figures/hard_in_lambda_L20C150N50.csv};
\addlegendentry{$\lambda=1.8\cdot 10^{-3}$};

\addplot[
 only marks,color=blue,
  mark=square,mark repeat=1,mark phase=0,mark size=.5mm,
  ]
table[
  header=false,x index=0,y index=2,
  ]
{matlab_figures/hard_in_lambda_L20C150N50.csv};
\addlegendentry{$\lambda=2\cdot 10^{-3}$};

\addplot[
 only marks,color=green!50!black,
  mark=o,mark repeat=1,mark phase=0,mark size=1mm
  ]
table[
  header=false,x index=0,y index=3,
  ]
{matlab_figures/hard_in_lambda_L20C150N50.csv};
\addlegendentry{$\lambda=2.2\cdot 10^{-3}$};

\end{axis}
\end{tikzpicture}
\caption{Influence of \JG{small differences in} $\lambda$ on storage policy under individual capacity constraints ($L=20$, $N=50$, $C=150$, $s=1$, $r=50$)}
\label{fig:varying_lambda}
\end{figure}
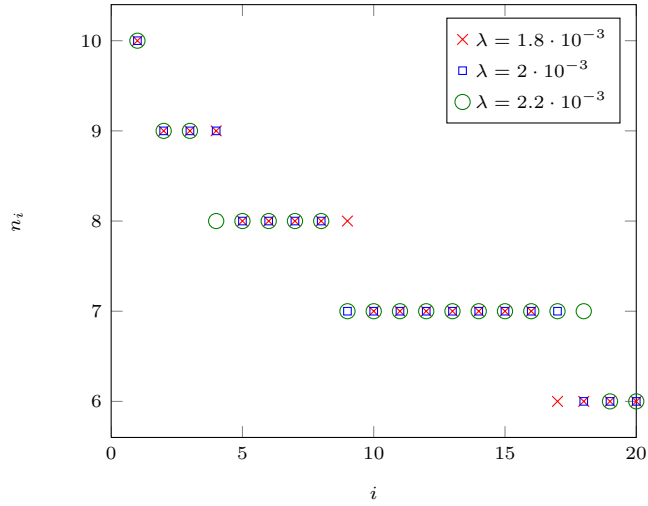

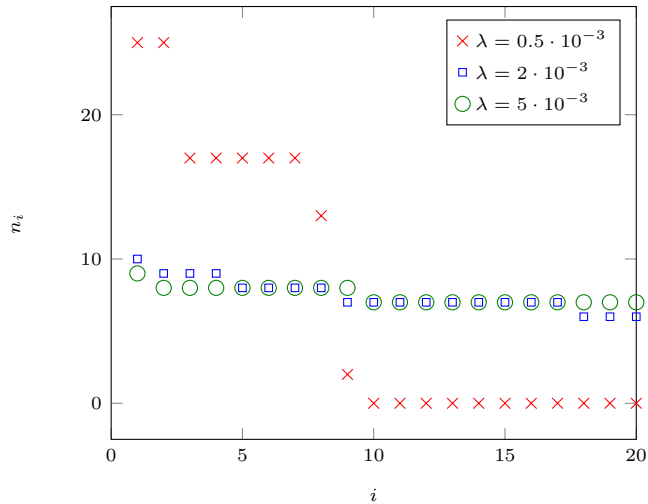
\begin{figure}
\centering
\begin{tikzpicture}
\begin{axis}[
  width=.7\linewidth,
  xlabel=$i$,ylabel=$n_i$, 
 xmin=0, xmax=20,
  font=\scriptsize,
  legend style={
        cells={anchor=west},
        legend pos=north east,
       font=\scriptsize,
    }
]

\addplot[
  only marks,color=red,
  mark=x,mark size=1mm
  ]
table[
  header=false,x index=0,y index=1,
  ]
{matlab_figures/hard_in_lambda_additional.csv};
\addlegendentry{$\lambda=0.5\cdot 10^{-3}$};

\addplot[
 only marks,color=blue,
  mark=square,mark repeat=1,mark phase=0,mark size=.5mm,
  ]
table[
  header=false,x index=0,y index=2,
  ]
{matlab_figures/hard_in_lambda_additional.csv};
\addlegendentry{$\lambda=2\cdot 10^{-3}$};

\addplot[
 only marks,color=green!50!black,
  mark=o,mark repeat=1,mark phase=0,mark size=1mm
  ]
table[
  header=false,x index=0,y index=3,
  ]
{matlab_figures/hard_in_lambda_additional.csv};
\addlegendentry{$\lambda=5\cdot 10^{-3}$};

\end{axis}
\end{tikzpicture}
\caption{
Influence of \JG{large differences in}  $\lambda$ on storage policy under individual capacity constraints ($L=20$, $N=50$, $C=150$, $s=1$, $r=50$)}
\label{fig:varying_lambda_b}
\end{figure}

\subsection{The optimal solution under an average capacity constraint}
Theorem~\ref{th:avg} provides an analytical expression for the optimal allocation probabilities under an average capacity constraint. It is not immediately clear from Theorem~\ref{th:avg} how $k_1$ and $k_2$ depend on, for instance, $x=\lambda\pi r^2$.
In Figures~\ref{fig:avgzipf2a} and~\ref{fig:avgzipf2b} we have illustrated the optimal allocation probabilities $\bar q_i$ under an average capacity constraint for various values of $r$. We observe that if $r$ is large, which means that the client can reach more caches within the range, then we store all files with equal probability $C/L$. It is intuitively clear that this minimizes the cache miss probability, since now all files can be retrieved with high probability. If $r$ is small and less caches can be reached, we will put priority, \ie higher $\bar q_i$, on the more popular files.

\subsection{Performance under individual and average cache capacity constraints}
Next we compare the miss probability of the optimal cache allocation under the individual cache capacity constraints with the miss probability under the average constraint. In Figure~\ref{fig:newa} we have depicted the cache miss probability as a function of $r$ for two values of $C$. From the discussion it should be clear that in the limit of large $r$ the cache miss probability under an average capacity constraint should approach zero. This is indeed reflected in Figure~\ref{fig:newa}. The individual capacity constraint, in stark contrast, results in a significant cache miss probability even at large $r$. The reason is that some files will not be stored at all and, therefore, a request for these files will always result in a cache miss.

Another interpretation of the significant improvement that is offered by allowing an average constraint, which means that whether the file is in the cache or not is probabilistic, is that different caches may have different files and that can help improve the performance and file diversity.

\subsection{Non-homogeneous distribution of base stations}

Here we argue that if the density of base stations $\lambda$ does not change very
rapidly, our analysis remains applicable but of course approximate.
Figure 3 gives performance as a function of $r$. What is important is that $x=\lambda\pi r^2$ is the only factor of influence. Therefore, our Figure 3 already gives some insight in the behavior as a function of $\lambda$. The storage policy under hard constraints is not influenced by the value of $x$ if $N=1$;
we simply store the most popular files. For $N>1$ we evaluate our dynamic programming policy. Figure~\ref{fig:varying_lambda} demonstrates the influence of $\lambda$ on the storage policy. The figure provides the number of fragments stored $n_i$ for file $i$ for various values of $\lambda$. The figure demonstrates that the policy is not changing much by small perturbations of $\lambda$. Hence, if the
density of base stations does not change much in space one can use a single policy everywhere without
much damage to the system performance. If the density of base stations changes significantly but not
too rapidly, as mentioned above, we expect that our results are still practically applicable.
Of course, as we demonstrate in the following \JG{Figure~\ref{fig:varying_lambda_b}}, the optimal policies for different densities of
base station distribution can be quite different.

\JGb{
\begin{figure}
\begin{center}
\includegraphics[width=.7\linewidth]{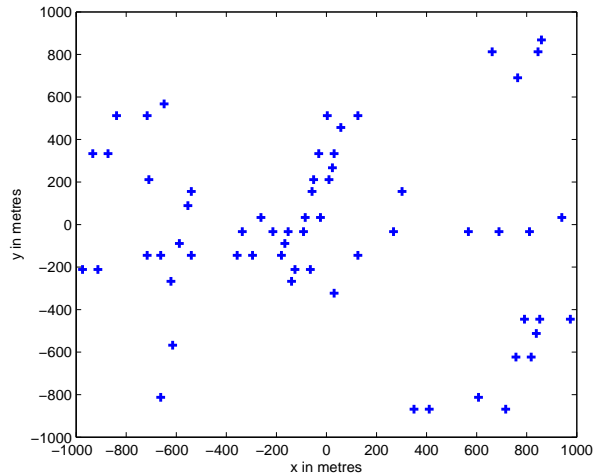}
\caption{Location of Base Stations from OpenMobileNetwork dataset.}
\label{fig:realdata}
\end{center}
\end{figure}
}

\JGb{
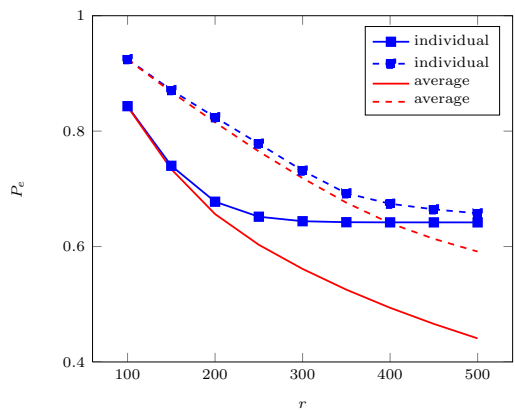
\begin{figure} 
\centering
\begin{tikzpicture}[scale=0.8]
\begin{axis}[
  width=.7\linewidth,
xlabel=$r$, ylabel=$P_e$, 
ymin = 0.4, ymax = 1,
font=\scriptsize,
legend style={
	cells={anchor=west},
	legend pos=north east,
	font=\scriptsize,
}
]

\addplot[
line width=.3mm,color=blue,
mark=square*
]
table[
header=false,x index=0, y index=2, col sep = comma
]
{berlin/compare_Berlin_PP.csv};
\addlegendentry{individual};

\addplot[
mark=square*,line width=.3mm,color=blue, dashed
]
table[
header=false, x index=0, y index=4, col sep = comma
]
{berlin/compare_Berlin_PP.csv};
\addlegendentry{individual};

\addplot[
mark=none,line width=.3mm,color=red
]
table[
header=false, x index=0, y index=1, col sep = comma
]
{berlin/compare_Berlin_PP.csv};
\addlegendentry{average};	

\addplot[
mark=none,line width=.3mm,color=red, dashed
]
table[
header=false, x index=0, y index=3, col sep = comma
]
{berlin/compare_Berlin_PP.csv};
\addlegendentry{average};
\end{axis}
\end{tikzpicture}
\caption{Cache miss probability in Berlin network and Poisson process. In solid line the performance of the Poisson process. In dashed line the performance of the Berlin network. ($C=10, L = 2000, \lambda = 1.8324\cdot 10^{-5}, s = 1$)}
\label{fig:berlin}
\end{figure}
}

\JGb{
Next we evaluate the performance of coded and uncoded strategies on the topology
of a real wireless network. Similar to our study in~\cite{altman2013coding} we have taken the positions of 3G base stations provided
by the OpenMobileNetwork project \cite{OpenMobileNetwork}. The base stations are situated
in the area $1.95 \times 1.74$ kms around the TU-Berlin campus. One can see the
positions of the base stations from the OpenMobileNetwork project in Figure~\ref{fig:realdata}.
We note that the base stations of the real network are more clustered then in a typical realization of a Poisson process, because they are
typically situated along roads. We will analyze the performance of our placement strategies (which are optimal for a Poisson network) on this non-Poisson topology. There are $62$ base stations in our dataset, corresponding to an average density of $\lambda = 1.8324 \cdot 10^{-5}$. We use this density to derive the optimal placement strategies under individual and average capacity constraints for various values of $r$. The results are depicted in Figure~\ref{fig:berlin}, which also includes the results for a Poisson network with the same density. We observe that the clustering increases the cache miss probability, but that our results on the Poisson model approximate the performance on the real data set quite well. The difference between the Poisson case and our dataset is smaller for individual capacity constraints than it is for average capacity constraints. We cannot explain this difference with our current results and suggest as future work to develop an insight into this behavior.
}

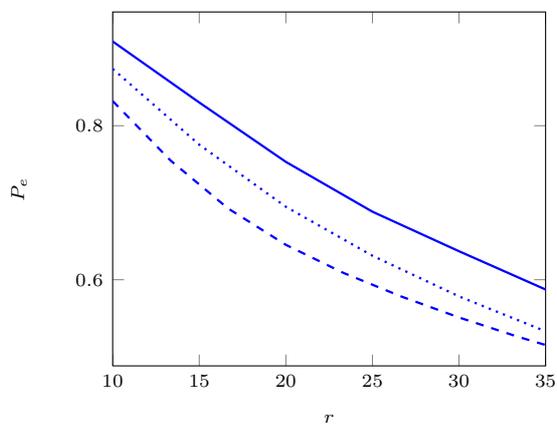
\begin{figure}
\centering
\begin{tikzpicture}
\begin{axis}[
  width=.6\linewidth,
  xlabel=$r$,ylabel=$P_e$, 
xmin=10,xmax=35,
  font=\scriptsize,
]

\addplot[
  mark=none,line width=.3mm,color=blue,dashed,
  mark=none
  ]
table[
  header=false,x index=0,y index=2,
  ]
{matlab_figures/hardsoftL2000C10.csv};

\addplot[
  mark=none,line width=.3mm,color=blue,solid,
  mark=none
  ]
table[
  header=false,x index=0,y index=1,
  ]
{new_matlab_figures/simC10.csv};

\addplot[
  mark=none,line width=.3mm,color=blue,dotted,
  mark=none
  ]
table[
  header=false,x index=0,y index=1,
  ]
{new_matlab_figures/approxC10.csv};

\end{axis}
\end{tikzpicture}
\caption{Dynamic scenario. In solid line performance of the \JGb{fully distributed} LRU policy. In dashed line the performance under optimal allocation of Section~\ref{sec:average}. In dotted line the performance of a single LRU cache with capacity equal to expected sum capacity of all caches that are within range of a client. ($C=10$, $L=2000$, $\lambda=2\cdot 10^{-3}$, $s=1$)}
\label{fig:simC10}
\end{figure}

\begin{figure}
\centering
\begin{tikzpicture}
\begin{axis}[
  width=.6\linewidth,
  xlabel=$r$,ylabel=$P_e$, 
xmin=10,xmax=35,
  font=\scriptsize,
]

\addplot[
  mark=none,line width=.3mm,color=blue,dashed,
  mark=none
  ]
table[
  header=false,x index=0,y index=2,
  ]
{matlab_figures/hardsoftL2000C50.csv};

\addplot[
  mark=none,line width=.3mm,color=blue,solid,
  mark=none
  ]
table[
  header=false,x index=0,y index=1,
  ]
{new_matlab_figures/simC50.csv};

\addplot[
  mark=none,line width=.3mm,color=blue,dotted,
  mark=none
  ]
table[
  header=false,x index=0,y index=1,
  ]
{new_matlab_figures/approxC50.csv};

\end{axis}
\end{tikzpicture}
\caption{Dynamic scenario. In solid line performance of the \JGb{fully distributed} LRU policy. In dashed line the performance under optimal allocation of Section~\ref{sec:average}. In dotted line the performance of a single LRU cache with capacity equal to expected sum capacity of all caches that are within range of a client. ($C=50$, $L=2000$, $\lambda=2\cdot 10^{-3}$, $s=1$)}
\label{fig:simC50}
\end{figure}
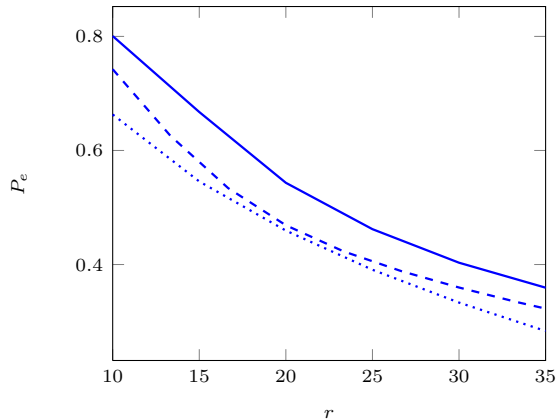

\subsection{Dynamic Setting}
Finally we consider the dynamic setting of Section~\ref{sec:dynamic}. In Figures~\ref{fig:simC10} and~\ref{fig:simC50} we have depicted the cache miss probability as a function of the connection range $r$ for cache capacities $C=10$ and $C=50$, respectively. In solid lines  we have depicted the performance of the \JGb{fully distributed} LRU policy. In dashed lines we have depicted the performance under the optimal allocation strategy of Section~\ref{sec:average}. Finally, we have depicted in dotted lines the performance of a `centralized' LRU policy, \ie we depict the performance of a single LRU cache with capacity equal to expected sum capacity of all caches that are within range of a client.

From Figures~\ref{fig:simC10} and~\ref{fig:simC50} it is clear that the performance of the \JGb{fully distributed} LRU policy is not too far from the performance of the optimal allocation strategy of Section~\ref{sec:average}. Another observation is that the performance difference between our distributed and the `centralized' LRU policy is small. Therefore, our \JGb{distributed} LRU-based policy with
caching in a closest storage device can be safely employed in practice for geographically distributed caching.

\section{Discussion}
\label{sec:disc}

In the current paper we have obtained structural insight into optimal storage allocation \JG{strategies} in a network of wireless caching devices
in a stochastic geometry. We have seen that for the design of geographically distributed caching devices it is better to use average than
per cache capacity constraint. We indicate that our model can be practically applied even for 
non-homogeneous distribution of base stations when the rate of density change is not too rapid.
We have also considered a dynamic setting for which we proposed a simple distributed LRU-based policy. We have shown that the performance of this LRU policy is not far from the optimal one, and consequently, this LRU-based policy can be safely employed in practice for geographically distributed caching. Part of the analysis in this paper considered the particular case that files consist of a single packet. In future work we will generalize this analysis. In addition we will consider the dynamic setting in more detail \JGb{by extending the model to include the latencies of fetching a file from a server and analyzing the overall file delivery latency.} In particular, the aim is to obtain a more fundamental insight into the behavior of LRU and others replacement policies in networks of wireless caches in a stochastic geometry setting. \JGb{A first step in understanding this behavior is to generalize the cache placement strategies from this paper to strategies that allow for a different (deterministic) placement of files in each of the caches. The optimal hit probability under such strategies can then serve as a baseline for online dynamic strategies. Also, it will enable to study non-homogeneous spatial Poisson processes for base station placement or more general placement in a natural way. Investigating different placement in each cache is part of our ongoing efforts as well as~\cite{chattopadhyay2016gibbsian}. Analyzing various online dynamic strategies can be approached by considering TTL caches~\cite{fofack2014performance}, which are known to cover many other strategies by carefully choosing the TTL distributions~\cite{dehghan2016utility}.}

%
%
%
 \section*{Acknowledgement}

 \JG{This work was performed while Xinwei Bai was visiting INRIA Sophia Antipolis in fall 2013. This work was supported in part by the Netherlands Organization for Scientific Research (NWO) grant 612.001.107.}

%
%
%
\bibliographystyle{IEEEtran}
\bibliography{IEEEabrv,abg_r1}

\end{document}